\documentclass[11pt,a4paper]{article}
\usepackage{fullpage}
\usepackage{microtype}
\usepackage[OT1]{fontenc}
\usepackage[tt=false]{libertine} 
\usepackage{amsmath}
\usepackage{amssymb}
\usepackage{amsthm}
\usepackage{mathtools} 
\usepackage[linesnumbered,boxed,ruled,vlined]{algorithm2e}
\usepackage{algpseudocode}
\usepackage{enumitem}
\usepackage{bm}
\usepackage{bbm}
\usepackage{float}
\usepackage{csquotes}
\usepackage{MnSymbol}

\usepackage{textpos}

\usepackage{todonotes}

\usepackage[margin=1cm]{caption} 
\usepackage{subfig}

\usepackage[thinlines]{easytable}

\patchcmd{\thebibliography}
  {\settowidth}
  {\setlength{\itemsep}{0pt plus 0.1pt}\settowidth}
  {}{}
\apptocmd{\thebibliography}
  {\small}
  {}{}

\usepackage{pgfplots}
\pgfplotsset{compat=1.16}
\usepackage{tikz}
\usetikzlibrary{arrows,arrows.meta,backgrounds,calc,fit,decorations.pathreplacing,decorations.markings,shapes.geometric}

\tikzstyle{internal} = [draw, fill, shape=circle]
\tikzstyle{external} = [shape=circle]
\tikzstyle{square}   = [draw, fill, rectangle]
\tikzstyle{triangle} = [draw, fill, regular polygon, regular polygon sides=3, inner sep=3pt]
\tikzstyle{pentagon} = [draw, fill, regular polygon, regular polygon sides=5, inner sep=2pt, minimum size=14pt]
\tikzset{every fit/.append style=text badly centered}

\usetikzlibrary{positioning,chains,fit,shapes,calc}
\usetikzlibrary{trees}
\usetikzlibrary{decorations.pathreplacing}
\usetikzlibrary{decorations.pathmorphing}
\usetikzlibrary{decorations.markings}
\tikzset{>=latex} 

\usepackage{ifthen}

\usepackage[bookmarks=true,hypertexnames=false,pagebackref]{hyperref}
\hypersetup{colorlinks=true, citecolor=blue, linkcolor=red, urlcolor=blue}
\usepackage{cleveref}

\usepackage[normalem]{ulem}

\usepackage{mleftright}

\usepackage{cool}
\Style{DSymb={\mathrm d},DShorten=true,IntegrateDifferentialDSymb=\mathrm{d}}

\newcommand{\defeq}{\coloneqq}

\newcommand{\Ppoly}{\textnormal{\textbf{P}}/_\textnormal{{poly}}}
\newcommand{\numP}{\#{\textnormal{\textbf{P}}}}

\newcommand{\VBP}{\textnormal{\textbf{VBP}}}
\newcommand{\VP}{\textnormal{\textbf{VP}}}

\newcommand{\omegamm}{\omega_{\mathrm{MM}}}
\newcommand{\omegadet}{\omega_{\mathrm{DET}}}

\let\det\undefined
\DeclareMathOperator{\det}{det}

\DeclareMathOperator{\wt}{wt}

\DeclareMathOperator{\holant}{holant}
\DeclareMathOperator{\npm}{pm}

\DeclareMathOperator{\inv}{inv}
\DeclareMathOperator{\Fchar}{char}

\newtheorem{theorem}{Theorem}

\newtheorem{lemma}[theorem]{Lemma}

\newtheorem{proposition}[theorem]{Proposition}
\newtheorem{corollary}[theorem]{Corollary}
\theoremstyle{definition}

\newtheorem{definition}[theorem]{Definition}

\theoremstyle{remark}
\newtheorem*{remark}{Remark}

\AddToHook{env/lemma/begin}{\crefalias{theorem}{lemma}}
\AddToHook{env/claim/begin}{\crefalias{theorem}{claim}}
\AddToHook{env/observation/begin}{\crefalias{theorem}{observation}}
\AddToHook{env/proposition/begin}{\crefalias{theorem}{proposition}}
\AddToHook{env/corollary/begin}{\crefalias{theorem}{corollary}}
\AddToHook{env/condition/begin}{\crefalias{theorem}{condition}}
\AddToHook{env/definition/begin}{\crefalias{theorem}{definition}}
\AddToHook{env/question/begin}{\crefalias{theorem}{question}}
\AddToHook{env/example/begin}{\crefalias{theorem}{example}}
\AddToHook{env/conjecture/begin}{\crefalias{theorem}{conjecture}}

\crefname{theorem}{Theorem}{Theorems}
\crefname{observation}{Observation}{Observations}
\crefname{claim}{Claim}{Claims}
\crefname{condition}{Condition}{Conditions}
\crefname{algorithm}{Algorithm}{Algorithms}
\crefname{property}{Property}{Properties}
\crefname{example}{Example}{Examples}
\crefname{fact}{Fact}{Facts}
\crefname{lemma}{Lemma}{Lemmas}
\crefname{corollary}{Corollary}{Corollaries}
\crefname{definition}{Definition}{Definitions}
\crefname{remark}{Remark}{Remarks}
\crefname{proposition}{Proposition}{Propositions}
\crefname{equation}{equation}{equations}
\crefname{enumi}{Case}{Case}
\creflabelformat{enumi}{(#2#1#3)}

\newcommand{\probname}[1]{\textnormal{\textsc{#1}}}

\newcommand{\prob}[3]{
\begin{center}
\noindent\parbox[t]{0.12\linewidth}{\raggedleft \textit{Name:}}\hfill\begin{minipage}[t]{0.83\linewidth}{#1}\end{minipage}

\medskip

\noindent\parbox[t]{0.12\linewidth}{\raggedleft \textit{Instance:}}\hfill\begin{minipage}[t]{0.83\linewidth}{#2}\end{minipage}

\medskip

\noindent\parbox[t]{0.12\linewidth}{\raggedleft \textit{Output:}}\hfill\begin{minipage}[t]{0.83\linewidth}
{#3}\end{minipage}
\end{center}
}

\makeatletter
\providecommand\@dotsep{5}
\def\listtodoname{Todo list}
\def\listoftodos{\@starttoc{tdo}\listtodoname}
\makeatother

\usetikzlibrary{shapes}
\definecolor{cbfp1}{RGB}{120,94,240}
\definecolor{cbfp2}{RGB}{220,38,127}
\definecolor{cbfp3}{RGB}{254,97,0}
\definecolor{cbfp4}{RGB}{255,176,0}
\definecolor{cbfp5}{RGB}{100,143,255}
\definecolor{verylightgray}{RGB}{216,216,216}

\newcommand{\gtdfv}{
\node [shape=diamond,fill=black,scale=0.3] (v) at (0,0) {};
\node [scale=0.2] (v1) at (0,1) {};
\node [scale=0.2] (v2) at (1,0) {};
\node [scale=0.2] (v3) at (0,-1) {};
\node [scale=0.2] (v4) at (-1,0) {};
}
\newcommand{\gtE}{\tikz[baseline=-.5*(height("$+$")-depth("$+$")),xscale=0.25,yscale=0.25]{\gtdfv
\path [draw=black!5,thick] (v) -- (v1) (v) -- (v2) (v) -- (v3) (v) -- (v4);
}}
\newcommand{\gtV}{\tikz[baseline=-.5*(height("$+$")-depth("$+$")),xscale=0.25,yscale=0.25]{\gtdfv
\path [draw=cbfp3,ultra thick] (v) -- (v1) (v) -- (v3);
\path [draw=black!5,thick] (v) -- (v2) (v) -- (v4);
}}
\newcommand{\gtH}{\tikz[baseline=-.5*(height("$+$")-depth("$+$")),xscale=0.25,yscale=0.25]{\gtdfv
\path [draw=black!5,thick] (v) -- (v1) (v) -- (v3);
\path [draw=cbfp3,ultra thick] (v) -- (v2) (v) -- (v4);
}}
\newcommand{\gtC}{\tikz[baseline=-.5*(height("$+$")-depth("$+$")),xscale=0.25,yscale=0.25]{\gtdfv
\path [draw=cbfp3,ultra thick] (v) -- (v1) (v) -- (v2) (v) -- (v3) (v) -- (v4);
}}
\newcommand{\gtF}{\tikz[baseline=-.5*(height("$+$")-depth("$+$")),xscale=0.25,yscale=0.25]{\gtdfv
\path [draw=black!5,thick] (v) -- (v1) (v) -- (v4);
\path [draw=cbfp3,ultra thick] (v) -- (v2) (v) -- (v3);
}}
\newcommand{\gtI}{\tikz[baseline=-.5*(height("$+$")-depth("$+$")),xscale=0.25,yscale=0.25]{\gtdfv
\path [draw=black!5,thick] (v) -- (v2) (v) -- (v3);
\path [draw=cbfp3,ultra thick] (v) -- (v1) (v) -- (v4);
}}
\newcommand{\gtITR}{\tikz[baseline=-.5*(height("$+$")-depth("$+$")),xscale=0.25,yscale=0.25]{\gtdfv
\path [draw=black!5,thick] (v) -- (v4) (v) -- (v3);
\path [draw=cbfp3,ultra thick] (v) -- (v1) (v) -- (v2);
}}
\newcommand{\gtIBL}{\tikz[baseline=-.5*(height("$+$")-depth("$+$")),xscale=0.25,yscale=0.25]{\gtdfv
\path [draw=black!5,thick] (v) -- (v2) (v) -- (v1);
\path [draw=cbfp3,ultra thick] (v) -- (v3) (v) -- (v4);
}}

\title{Planar Perfect Matching Counting is as Hard as Determinants}
\date{}

\author{Radu Curticapean\thanks{University of Regensburg and IT University of Copenhagen. Funded by the
European Union (ERC, CountHom, 101077083). Views and opinions expressed are however those of the author(s) only and do not necessarily reflect those of the European Union or the European Research Council Executive Agency.}
\and Jiaheng Wang\thanks{University of Regensburg and University of Helsinki. Supported by the ERC project CountHom and the Helsinki Institute of Information Technology (HIIT).}}

\begin{document}

\maketitle

\begin{abstract}
In the 1960s, Fisher, Kasteleyn and Temperley designed an ingenious algorithm for computing the partition function of the \emph{dimer model}, or equivalently, for counting perfect matchings in edge-weighted planar graphs (\textit{Philos.~Mag.} 1961; \textit{J.~Mathematical~Phys.} 1963). 
This \emph{FKT algorithm} later became the foundation for Valiant's \emph{holographic algorithms} (FOCS 2004; \textit{SIAM~J.~Comput.} 2008), which motivated the study of counting problems under the Holant framework.
Combined with an algorithm by Yuster (FOCS 2008), the FKT algorithm allows us to count edge-weighted perfect matchings in planar $n$-vertex graphs with $\tilde{O}(n^{\omega/2})$ arithmetic operations, where $\omega<2.372$ is the matrix multiplication exponent. 

We prove a corresponding lower bound: Over algebraic circuits and other sufficiently strong computational models, perfect matchings in edge-weighted $n$-vertex planar graphs $G$ cannot be counted in $O(n^{\omega/2-\epsilon})$ arithmetic operations. This confirms the optimality of Yuster's algorithm.
Our bound holds even when $G$ is an edge-weighted square grid.  
\end{abstract}

\section{Introduction}

The complexity class \numP{} introduced by Valiant \cite{DBLP:journals/siamcomp/Valiant79} shows a fundamental separation between decision and counting problems, most prominently for the \emph{perfect matching} problem: 
while the existence of perfect matchings in graphs can be determined in polynomial time~\cite{MR177907}, counting them is \numP{}-complete~\cite{DBLP:journals/siamcomp/Valiant79}.
Given this intractability, the Fisher--Kasteleyn--Temperley (FKT) algorithm \cite{MR153427,MR136398} stands as a remarkable anomaly, as it allows us to count perfect matchings in \emph{planar} graphs $G$ in polynomial time.
The FKT algorithm achieves this by introducing carefully chosen signs into the adjacency matrix of $G$ and using an ingenious cancellation property to ensure that the determinant of the resulting matrix counts (pairs of) perfect matchings in $G$.
The FKT method also works when the input graph $G$ is edge-weighted and matchings are weighted by the product of involved edge-weights; we denote this computational problem as \probname{\#PlanarPM}.

Beyond its implications in statistical physics such as \emph{Exactly Solved Models}~\cite{Baxter82,Welsh_1993}, the FKT algorithm is the foundation of Valiant's \emph{holographic algorithms}, which are polynomial-time algorithms for problems that appear to evade classical techniques in algorithm design~\cite{ DBLP:conf/focs/Valiant06,DBLP:journals/siamcomp/Valiant08,DBLP:journals/mst/CaiFGW22,DBLP:journals/siamcomp/Backens21}.
These algorithms gave birth to the ongoing project of classifying the complexity of counting problems under the \emph{Holant} framework, which led to a plethora of research papers, monographs and a textbook over the last two decades, e.g., \cite{DBLP:journals/jcss/CaiL11,DBLP:journals/cc/HuangL16,DBLP:conf/focs/0001C20,DBLP:journals/mst/CaiLX20,DBLP:journals/mst/CaiFGW22,CaiChenBook}. 
In most of these results, the planar cases are solved by an algorithmic template that ultimately involves the FKT algorithm.
Indeed, every tractable planar \emph{\#CSP} can be solved along these lines~\cite{DBLP:journals/siamcomp/CaiF22};
this was discovered while studying to which extent the FKT algorithm is \emph{universal} for counting problems on planar graphs. (It is not universal for \emph{Holant problems}~\cite{DBLP:journals/mst/CaiFGW22}.)
In this paper, we study a different question---

\subsection*{Is the FKT algorithm optimal?}

Given the fundamental role of the problem \probname{\#PlanarPM} in statistical physics and theoretical computer science, e.g., in the form of Holant problems, it is important to determine its precise complexity. 
The FKT algorithm readily gives an upper bound:
On an $n$-vertex planar graph, the running time of this algorithm is dominated by the evaluation of a particular $n\times n$ determinant $\det(A)$. There are several options for computing $\det(A)$: 
\begin{itemize}
    \item Algorithms for generic determinants allow us to evaluate $\det(A)$ with $\tilde{O}(n^{\omega})$ arithmetic operations, where $\omega <2.372$ is the \emph{matrix multiplication exponent}.\footnote{
To be more precise, but a bit technical: We can further replace the matrix multiplication constant with the\emph{matrix determinant constant} $\omegadet$ since the algorithm only asks for the determinant. 
The relation between $\omegadet$ and $\omegamm$ is discussed later at the end of the introduction. 
}
    \item However, the matrix $A$ inherits a great deal of structure from the planarity of the input graph $G$. This can be exploited, e.g., using the \emph{nested dissection method}~\cite{MR388756,MR526496} and additional insights~\cite{DBLP:conf/soda/Wilson97,DBLP:journals/algorithmica/MuchaS06,DBLP:conf/soda/YusterZ07,DBLP:conf/focs/Yuster08}, to evaluate $\det(A)$ in $\tilde{O}(n^{\omega/2})$ arithmetic operations.
\end{itemize}

The exponent $\omegamm/2$ is a natural barrier for the nested dissection method, since it relies on planar separators, which may require $O(\sqrt{n})$ vertices in the worst case.
Consequently, any attempt at further improving the running time would likely need to bypass the separator-based approach.

Note however that this algorithmic strategy might already be optimal, and we merely do not know yet: Our current knowledge permits the possibility of $\omegamm =2$, which would imply an optimal exponent of $\omegamm/2 = 1$.
Conversely, an $\Omega(n^{1+\epsilon})$ lower bound on \probname{\#PlanarPM} would directly translate into $\omegamm >2$. Such a lower bound however appears to be out of reach, even conditioned on the strong exponential-time hypothesis \cite{ImpagliazzoP01}, the APSP conjecture \cite{Williams18}, or other popular fine-grained conjectures \cite{WilliamsW18}. 

\subsection*{Our results}
Although the current fine-grained complexity landscape offers little explanation for the $\omega/2$ barrier for the problem \probname{\#PlanarPM}, we can still obtain a meaningful statement about its complexity by relating it to \emph{generic} determinants: 
By showing that generic $n \times n$ determinants can be reduced to \probname{\#PlanarPM} on grid graphs of side-length $O(n)$, we prove the following lower bound:

\begin{theorem} \label{thm:main}
No algorithm solves \probname{\#PlanarPM} on planar $n$-vertex edge-weighted graphs with $O(n^{\omega/2-\epsilon})$ arithmetic operations for any $\epsilon>0$, even when the input graph is an edge-weighted grid. 
Here, $\omega=\omegadet$ is the matrix determinant constant.
\end{theorem}

Consequently, over $\mathbb{R}$, the $\tilde{O}(n^{\omegamm/2})$ algorithms by Yuster and Zwick \cite{DBLP:conf/soda/YusterZ07,DBLP:conf/focs/Yuster08} are asymptotically optimal (up to polylogarithmic factors) in computational models with $\omegadet=\omegamm$. This includes standard algebraic models that support the Baur--Strassen Theorem~\cite{DBLP:journals/tcs/BaurS83}, e.g., algebraic circuits and algebraic branching programs.

Regarding the role of edge-weights, we note that the graphs constructed in the proof of \cref{thm:main} contain edges of negative weight. It is possible to remove such edges, but the running time overhead caused in the reduction would lead to a weaker bound.
Over prime fields $\mathbb{F}_p$, a modified stretching/thickening argument \cite{MR1049758} can be implemented with only a constant-factor overhead, yielding the following theorem:

\begin{theorem} \label{thm:main-modulo}
For any fixed prime $p\geq 2$ and $\epsilon>0$, no algorithm solves \probname{\#PlanarPM} over $\mathbb{F}_p$ on unweighted planar $n$-vertex graphs with $O(n^{\omega/2-\epsilon})$ arithmetic operations.
Here, $\omega=\omegadet$ is the matrix determinant constant.
\end{theorem}

\subsection*{Connections to Algebraic Complexity Theory}
Determinants are complete for the algebraic complexity class $\VBP$, which captures computations by polynomial-sized \emph{algebraic branching programs} or \emph{skew} arithmetic circuits.\footnote{
A skew circuit involves at least one input gate at each multiplication gate. 
The related \emph{weakly-skew} circuits and algebraic branching programs turn out to have \emph{exactly} the same computational power \cite[Section 2]{DBLP:conf/issac/KaltofenK08}; see also \cite[Remark 2.22]{DBLP:journals/corr/abs-2406-06217}.}
Dropping skewness yields the class $\VP$, the algebraic version of $\Ppoly$ \cite{DBLP:conf/stoc/Valiant79a}. 
It is conjectured that $\VBP\neq\VP$.

It is known that determinants capture skew circuits and $\VBP$ with low overhead.
In the following, we write $\det_m$ for the determinant of an $m\times m$ matrix with generic indeterminates $x_{ij}$; this is a multivariate polynomial of degree $m$. The following is known:
\begin{itemize}
\item The determinant $\det_m$ admits $O(m^4)$-size skew circuits, e.g., by Berkowitz's method \cite{DBLP:journals/ipl/Berkowitz84,MR7497} or the beautiful combinatorial approach of Mahajan and Vinay \cite{DBLP:journals/cjtcs/MahajanV97}.  
\item Every function with a skew circuit of size $m$ is a projection of $\det_{m+1}$ \cite{toda1992classes,DBLP:journals/jc/MalodP08}.
\end{itemize}

Since the FKT algorithm establishes a reduction from \probname{\#PlanarPM} to determinants, one may wonder whether the converse reduction is also possible.
This seems plausible \emph{a priori}, because the determinant $\det(A)$ arising in the FKT algorithm is, up to signs, the adjacency matrix of an arbitrary edge-weighted planar graph $G$.
Indeed, Flarup, Koiran and Lyaudet \cite{DBLP:conf/isaac/FlarupKL07} already established 20 years ago that every polynomial $p$ with size-$n$ skew circuits is a projection of \probname{\#PlanarPM} with $O(n^2)$-vertex planar graphs. In this projection, edge-weights are univariate linear functions in the variables of $p$. 
Composing this result with the $\VBP$-completeness reduction of the determinant, one can express $\det_m$ as \probname{\#PlanarPM} on a planar graph on $O(m^8)$ vertices. 
Thus, up to polynomial factors, the determinant and \probname{\#PlanarPM} express the same polynomials through projections.

To study the \emph{fine-grained complexity} of \probname{\#PlanarPM} however, the $O(m^8)$ blowup is prohibitive.
Towards our main theorem, we establish a reduction with the \emph{optimal} blowup of $O(m^2)$.
This bound is indeed optimal for the trivial reason that $n$-vertex planar graphs have $O(n)$ edges and therefore require $\Omega(m^2)$ vertices to capture the $m^2$ variables in $\det_m$ by distinct edges.

\begin{theorem}[Optimal Expressiveness] \label{lem:reduction-grid}
Let $m\in\mathbb{N}$ and let $\bm{X}=(x_{ij})$ be an $m\times m$ matrix with indeterminates.
Then there is an edge-weighted $3m\times 5m$ grid graph $G_{\bm{X}}$ with edge weights $\{-1,0,1\}\cup\{1-x_{ij}: i,j\in[m]\}$ such that, as polynomials, 
\[
\det(\bm{X})\equiv\npm(G_{\bm{X}}). 
\]
\end{theorem}

We remark that the negative edge-weights introduced in \cref{lem:reduction-grid} are necessary, as the determinant can be negative even on $0$-$1$ matrices, while the plain \emph{unweighted} count of perfect matchings in a graph is of course always nonnegative.

\subsection*{Proof Overview}
Our proof is fully elementary and self-contained. In the proof, we use a minimal subset of the vast theory of \emph{Holant problems}, which we outline below. In machine learning and physics, such problems are also known as (contractions of) \emph{tensor networks} \cite{CLSWW2024,Huggins_2019,PhysRevLett.115.180405}.

\paragraph{Holant problems.}
Holant problems count weighted Boolean assignments to the edges of a graph. The weight of an assignment is determined by local factors contributed by the vertices.
More formally, consider a graph $G=(V,E)$ equipped with a set of functions $\mathcal{F}=\{f_v\}$, so-called \emph{signatures}, such that $f_v:\{0,1\}^{\deg(v)}\to\mathbb{C}$ on vertex $v$ obtains as input an assignment to the edges incident with $v$.
Any $0$-$1$ edge assignment $\sigma$ gives rise to a weight, which is the product of all $f_v$ under $\sigma$. 
The \emph{Holant} on $(G;\mathcal{F})$ is the sum of these weights over all possible $\sigma$: 
\[
\holant(G;\mathcal{F})=\sum_{\sigma:E\to\{0,1\}}\prod_{v\in V}f_v(\sigma|_{E(v)}).
\]

As an example, the problem of counting perfect matchings in a graph $G$ can be viewed as a Holant problem: Associate to each vertex $v\in V(G)$ the signature $f_v: \{0,1\}^{\deg(v)} \to \{0,1\}$ with 
\[
f_v(a)=1 \quad \text{if and only if the Hamming weight of $a$ is $1$}.
\]

Counting perfect matchings with weights can also be expressed as a Holant problem: It suffices to subdivide edges once and place signatures on the subdivision vertices that capture their edge-weight.\footnote{On the subdivision vertex $s$ induced by an edge of weight $w$, place $f(00)=1$ and $f(01)=f(10)=0$ and $f(11)=w$.}
When $G$ is bipartite, counting weighted perfect matchings in $G$ amounts to evaluating the \emph{permanent} of the bi-adjacency matrix of $G$.

\begin{figure}[t]
\centering
\begin{tikzpicture}
\clip (-3,-2.62) rectangle (3,2.62);
\draw [black,opacity=0] (-3,-2.62) rectangle (3,2.62);

\begin{scope}[shift={(-2.5,-2.5)}]

\foreach \x in {1,...,4}
\foreach \y in {1,...,4} {
    \node [diamond, draw=black, inner sep=2pt] (g\x\y) at ({\x},{\y}) {};
}

\foreach \x in {1,...,4} {
    \node [rectangle, draw=black, inner sep=2.828pt] (g\x 0) at ({\x},{0}) {};
    \node [circle, draw=black, inner sep=2pt] (g0\x) at ({0},{\x}) {};
    \node [circle, draw=black, inner sep=2pt] (g\x 5) at ({\x},{5}) {};
    \node [rectangle, draw=black, inner sep=2.828pt] (g5\x) at ({5},{\x}) {};
}

\def\chippathx#1#2{
\pgfmathparse{int({#1})} \xdef\finalx{\pgfmathresult}
\pgfmathparse{int({#2})} \xdef\y{\pgfmathresult}
\foreach \x in {0,...,4} {
    \pgfmathparse{int(1+\x)} \xdef\xx{\pgfmathresult}
    \pgfmathparse{int(100-(\x<\finalx)*100)} \xdef\bval{\pgfmathresult}
    \path [draw=cbfp3!\bval!verylightgray, opacity=0.75, very thick] (g\x\y) edge node []  {} (g\xx\y);
}
}
\def\chippathy#1#2{
\pgfmathparse{int({#2})} \xdef\finaly{\pgfmathresult}
\pgfmathparse{int({#1})} \xdef\x{\pgfmathresult}
\foreach \y in {0,...,4} {
    \pgfmathparse{int(1+\y)} \xdef\yy{\pgfmathresult}
    \pgfmathparse{int((\y<\finaly)*100)} \xdef\bval{\pgfmathresult}
    \path [draw=cbfp3!\bval!verylightgray, opacity=0.75, very thick] (g\x\y) edge node []  {} (g\x\yy);
}
}
\def\flipvtx#1#2{
\node [diamond, draw=black, inner sep=2pt, fill=cbfp1!50] () at (#1,#2) {};
\chippathx{#1}{#2}
\chippathy{#1}{#2}
}

\flipvtx{1}{4}
\flipvtx{2}{1}
\flipvtx{3}{3}
\flipvtx{4}{2}

\end{scope}

\end{tikzpicture}
\caption[]{A blue diamond indicates a marshaller. Orange lines indicate light beams. The circles and squares indicate sentinel vertices that are useful in simplifying the Holant construction. A constellation is valid iff each diamond is in state \gtE, \gtH, \gtV, \gtC, or \gtF, no circle has an incident beam, and each square has an incident beam.
}
\label{fig:per-from-grid}
\textcolor{lightgray}{\hrule}
\end{figure}
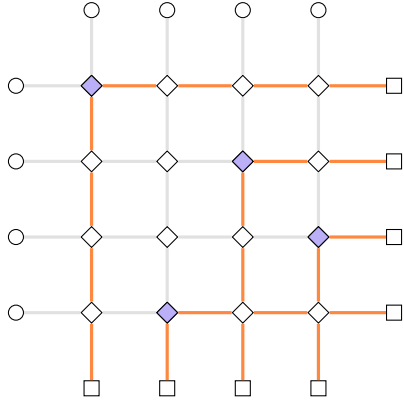

\paragraph{Permanents as Holants on grids.}
The above transformation from permanents into Holant problems directly mirrors the structure of the input graph into the Holant problem instance. 
Recent work on the complexity of multilinear forms~\cite{DBLP:journals/eccc/BrandCKLOS026} established a more input-agnostic transformation: As shown in~\cite{DBLP:journals/eccc/BrandCKLOS026}, the $n\times n$ permanent can be expressed as the Holant of an $n \times n$ \emph{grid graph} $G$. 
In this construction, the vertices of the grid $G$ correspond directly to the entries of $A$, and appropriate signatures ensure that the Holant counts combinatorial structures that correspond directly to row-column permutations $\pi \in S_n$, each weighted by $\prod_{i=1}^n a_{i,\pi i}$.

The structures counted in the Holant problem have a clean combinatorial interpretation, see also \Cref{fig:per-from-grid}: They are constellations of $n$ ``marshallers'' and $n$ vertical and $n$ horizontal ``light beams'' on the grid, subject to the following constraints:\footnote{This construction is very similar to a \emph{non-attacking rook placement}. The marshallers could be viewed as the main characters in the construction, but the beams are crucial to ensure a clean formulation as a Holant problem.}
\begin{enumerate}
\item Each marshaller sends a horizontal beam to the right and a vertical beam downwards. It contributes the factor $a_{i,j}$ to the weight when located at position $(i,j)$.
\item The beams of any marshaller may not hit another marshaller.
\item Beams may cross without any consequences.
\end{enumerate}

It is easy to see that the valid constellations of $n$ marshallers on the $n\times n$ grid correspond bijectively to permutations. Indeed, every such constellation represents a permutation matrix $P_\pi$ overlaid on $A$.
The weight of a constellation is the product of weights from marshallers, and thus equal to $\prod_{i=1}^n a_{i,\pi i}$.

The relevant constraints on marshallers and beams can be enforced straightforwardly with appropriate Holant signatures: If an edge is assigned $1$, it is understood as carrying a beam.

\paragraph{Determinants as Holants on grids.}
In this paper, we observe that the constellations described above also capture the \emph{determinant}: 
Indeed, in the constellation corresponding to a permutation $\pi$, beam crossings occur \emph{precisely} at the inversions of $\pi$. Thus, to express the determinant as a Holant problem over a grid graph, one only needs to introduce a negative sign for crossing beams. In other words, beam crossings are no longer consequence-free; they are still allowed, but they must be accounted for by a negative sign. 
On the level of Holants, this is easily reflected by modifying the signature output for the state $\gtC$ from $+1$ to $-1$.

The modified signature $f$ obtained this way satisfies the \emph{Matchgate Identities}~\cite{DBLP:journals/siamcomp/Valiant02,DBLP:journals/tcs/Valiant02,DBLP:journals/ijsi/CaiC07,DBLP:journals/mst/CaiCL09,DBLP:journals/toc/CaiG14}, which implies the existence of a planar graph $H$ with $O(1)$ vertices and four dangling edges that simulates $f$: For every subset $S$ of dangling edges, the signature value $f(S)$ equals the weighted number of perfect matchings in the graph $H-S$ obtained from $H$ by deleting $S$ and all incident vertices.
Cai and Gorenstein~\cite{DBLP:journals/toc/CaiG14} give an elegant explicit construction of $H$, which was also used, e.g., in \cite{DBLP:conf/sosa/CurticapeanX22}. 
The final planar graph $G_{\bm{X}}$ to simulate the determinant is then obtained by replacing every vertex of signature $f$ in the grid by a copy of the gadget $H$.
The resulting graph again fits neatly into a grid. 

\section{Preliminaries} 
\subsection{Perfect Matchings} 
A perfect matching of a graph $G=(V,E)$ is an edge subset $M \subseteq E$ that every vertex $v\in V$ appears exactly once in $M$. 
On input an edge-weighted graph $G=(V,E,w)$, the computational task \probname{\#PlanarPM} asks for a weighted count of perfect matchings in $G$, where each matching is weighted by the product of its edge weights. 
Formally:
\prob{\probname{\#PlanarPM}}{A simple undirected planar graph $G=(V,E,w)$ with edge weights $w:E\to\mathbb{C}$.}{Weighted sum of all perfect matchings $M\subseteq E$: 
\[
\npm(G)\defeq\sum_{M}\prod_{e\in M}w(e).
\]
}

\subsection{Determinant versus Matrix Multiplication} 
The matrix multiplication constant $\omegamm$ is defined as the smallest real number such that $n$-by-$n$ matrix multiplication admits an algorithm using $n^{\omegamm+o(1)}$ arithmetic operations. 
The matrix determinant constant $\omegadet$ is defined analogously. 
A classical result states that $\omegamm=\omegadet$ under standard algebraic models, and both directions are highly non-trivial:
\begin{enumerate}
    \item Bunch and Hopcroft showed $\omegamm\geq\omegadet$ by means of a carefully-designed LU decomposition and analysis~\cite{MR248973,MR331751}. 
    \item Baur and Strassen established $\omegamm\leq\omegadet$ by computing the partial derivatives together with the original function, and then applying Cramer's Rule \cite{DBLP:journals/tcs/BaurS83}. 
\end{enumerate}
However, the relation between matrix multiplication and matrix determinant is not known beyond algebraic models, e.g., in terms of bit complexity. 
It is worth remarking that Yuster's algorithm runs in $O(n^{\omegamm/2+1})$ bit operations. 

\begin{remark}
There is no need to specify the underlying field in \Cref{lem:reduction-grid}, because the matrix entries are treated as indeterminates. 
\Cref{thm:main} therefore works for every field; one just needs to specify the field for defining the matrix determinant constant. 

On the other hand, it is a major open problem whether these constants are field-independent, i.e., whether $\omega(F)$ agrees for all fields $F$.
All existing techniques for bounding $\omega$ are field-independent, 
and it is known that $\omegamm(F)=\omegamm(F')$ if $\Fchar(F)=\Fchar(F')$
\cite[Theorem 2.8]{DBLP:journals/siamcomp/Schonhage81}, which also suggests field independency.    
\end{remark}

\section{The reduction} \label{sec:proof}

We first describe the construction of individual gadgets and then show how to compose the gadgets to obtain the overall reduction.

\subsection{Gadgets}
Our reduction is based around the planar gadget $H$ shown in \Cref{fig:matchgate}, which has $6$ vertices, $7$ internal edges, and $4$ external \emph{dangling edges}, which we denote by $t,r,b,\ell$, going clockwise from the top. 
If a dangling edge is assigned $0$, the attaching vertex has to be \emph{matched within} the gadget, otherwise it must be left \emph{unmatched within} the gadget. 
Assignments are represented as $4$-bit strings in the order $t,r,b,\ell$, e.g., the string $0110$ represents the assignment $t=\ell=0,r=b=1$. 
The weight of $H$ on a given assignment to the dangling edges is the weighted count of perfect matchings in $H-S$, where $S$ is the set of endpoints of all dangling edges that are assigned $1$.

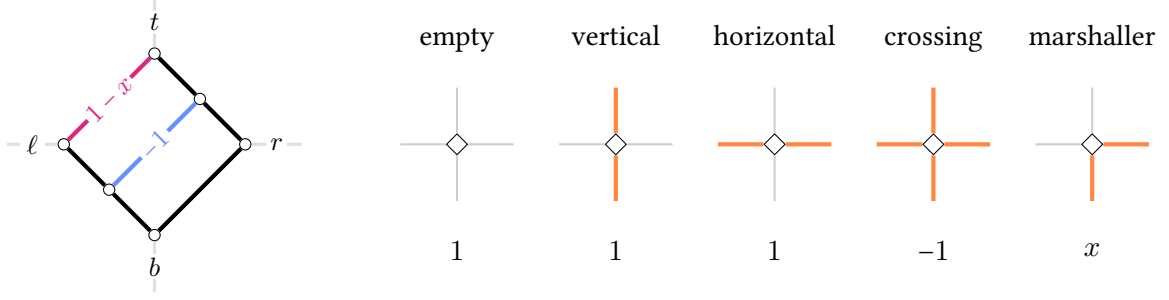
\begin{figure}[t]
\centering
\begin{tikzpicture}
\clip (-8.1,-2) rectangle (7.3,2);
\draw [black,opacity=0] (-8.1,-2) rectangle (7.3,2);

\begin{scope}[shift={(-6,0)}]
\node [circle, draw=black, fill=white, inner sep=1.5pt] (t) at (0,1.2) {};
\node [circle, draw=black, fill=white, inner sep=1.5pt] (r) at (1.2,0) {};
\node [circle, draw=black, fill=white, inner sep=1.5pt] (b) at (0,-1.2) {};
\node [circle, draw=black, fill=white, inner sep=1.5pt] (l) at (-1.2,0) {};
\node [circle, draw=black, fill=white, inner sep=1.5pt] (tr) at (0.6,0.6) {};
\node [circle, draw=black, fill=white, inner sep=1.5pt] (bl) at (-0.6,-0.6) {};

\path [draw=lightgray, opacity=0.5, very thick] 
(l) edge node [font=\small, midway, inner sep=0.5pt, minimum size=0.2pt, circle, align=center, fill=white,opacity=1] {$\ell$} +(-0.75,0);
\path [draw=lightgray, opacity=0.5, very thick] 
(t) edge node [font=\small, midway, inner sep=0.5pt, minimum size=0.2pt, circle, align=center, fill=white,opacity=1] {$t$} +(0,0.75);
\path [draw=lightgray, opacity=0.5, very thick] 
(r) edge node [font=\small, midway, inner sep=0.5pt, minimum size=0.2pt, circle, align=center, fill=white,opacity=1] {$r$} +(0.75,0);
\path [draw=lightgray, opacity=0.5, very thick] 
(b) edge node [font=\small, midway, inner sep=0.5pt, minimum size=0.2pt, circle, align=center, fill=white,opacity=1] {$b$} +(0,-0.75);

\path [ultra thick, draw=cbfp2] (l) edge node [font=\sffamily\small, sloped, midway, inner sep=0.5pt, minimum size=0.2pt, circle, align=center, fill=white,opacity=1] {\textcolor{cbfp2}{$1-x$}} (t);
\path [ultra thick, draw=cbfp5] (tr) edge node [font=\sffamily\small, sloped, midway, inner sep=0.5pt, minimum size=0.2pt, circle, align=center, fill=white,opacity=1] {\textcolor{cbfp5}{$-1$}} (bl);
\path [ultra thick] (b) edge node [] {} (r);
\path [ultra thick] (t) edge node [] {} (tr);
\path [ultra thick] (tr) edge node [] {} (r);
\path [ultra thick] (b) edge node [] {} (bl);
\path [ultra thick] (bl) edge node [] {} (l);
\end{scope}

\def\danglevtx#1#2#3#4#5#6#7#8{
\node [diamond, draw=black, inner sep=2pt] (t#3#4#5#6) at ({#1},{#2}) {};
\pgfmathparse{int(#3*100)} \xdef\bval{\pgfmathresult}
\ifthenelse{\equal{#3}{1}}{\def\tckness{ultra thick}}{\def\tckness{thick}}
\path [draw=cbfp3!\bval!lightgray, opacity=0.75, \tckness] 
(t#3#4#5#6) edge node [] {} +(-0.75,0);
\pgfmathparse{int(#4*100)} \xdef\bval{\pgfmathresult}
\ifthenelse{\equal{#4}{1}}{\def\tckness{ultra thick}}{\def\tckness{thick}}
\path [draw=cbfp3!\bval!lightgray, opacity=0.75, \tckness] 
(t#3#4#5#6) edge node [] {} +(0,0.75);
\pgfmathparse{int(#5*100)} \xdef\bval{\pgfmathresult}
\ifthenelse{\equal{#5}{1}}{\def\tckness{ultra thick}}{\def\tckness{thick}}
\path [draw=cbfp3!\bval!lightgray, opacity=0.75, \tckness] 
(t#3#4#5#6) edge node [] {} +(0.75,0);
\pgfmathparse{int(#6*100)} \xdef\bval{\pgfmathresult}
\ifthenelse{\equal{#6}{1}}{\def\tckness{ultra thick}}{\def\tckness{thick}}
\path [draw=cbfp3!\bval!lightgray, opacity=0.75, \tckness] 
(t#3#4#5#6) edge node [] {} +(0,-0.75);
\node [] () at ({#1},{#2-1.4}) {#7};
\node [] () at ({#1},{#2+1.4}) {{#8\vphantom{fhligpy}}};
}

\danglevtx{-2}{0}{0}{0}{0}{0}{$1$}{empty}
\danglevtx{0.1}{0}{0}{1}{0}{1}{$1$}{vertical}
\danglevtx{2.2}{0}{1}{0}{1}{0}{$1$}{horizontal}
\danglevtx{4.3}{0}{1}{1}{1}{1}{$-1$}{crossing}
\danglevtx{6.4}{0}{0}{0}{1}{1}{$x$}{marshaller}

\end{tikzpicture}
\caption{The planar $H$-gadget.}
\label{fig:matchgate}

\textcolor{lightgray}{\hrule}
\end{figure}

\begin{definition}[states of a gadget, see {\Cref{fig:matchgate}}]
\[
\begin{matrix}
\gtE & \defeq & 0000 & \text{(empty)}\\
\gtV & \defeq & 1010 & \text{(vertical)}\\
\gtH & \defeq & 0101 & \text{(horizontal)}\\
\gtC & \defeq & 1111 & \text{(crossing)}\\
\gtF & \defeq & 0110 & \text{(marshaller)}
\end{matrix}
\]
The $11$ of $16$ remaining states from $\{0,1\}^4$ are considered \emph{invalid}.
\end{definition}

\begin{lemma} \label{lem:gadget-weight}
The weight of an $H$-gadget is
\begin{align*}
1,      &\quad   \text{if $H$ is \gtE, \gtV or \gtH,}\\
-1,     &\quad   \text{if $H$ is \gtC,}\\
x,      &\quad   \text{if $H$ is \gtF, and}\\
0,      &\quad   \text{if $H$ is invalid.}
\end{align*}
\end{lemma}

\begin{figure}[t]
\centering
\begin{tikzpicture}
\clip (-8,-5.7) rectangle (8,1.1);
\draw [black,opacity=0] (-8,-5.7) rectangle (8,1.1);

\newcommand{\gadgetstate}[9]{
\begin{scope}[shift={(#1)}]
\node [circle, draw=black, fill=white, inner sep=1.2pt] (t) at (0,0.7) {};
\node [circle, draw=black, fill=white, inner sep=1.2pt] (r) at (0.7,0) {};
\node [circle, draw=black, fill=white, inner sep=1.2pt] (b) at (0,-0.7) {};
\node [circle, draw=black, fill=white, inner sep=1.2pt] (l) at (-0.7,0) {};
\node [circle, draw=black, fill=white, inner sep=1.2pt] (tr) at (0.35,0.35) {};
\node [circle, draw=black, fill=white, inner sep=1.2pt] (bl) at (-0.35,-0.35) {};

\path [draw=lightgray, opacity=0.5, thick] 
(l) -- +(-0.4,0)
(t) -- +(0,0.4)
(r) -- +(0.4,0)
(b) -- +(0,-0.4)
(l) -- (t) -- (tr) -- (r) -- (b) -- (bl) -- (l)
(tr) -- (bl)
;

\ifnum#2=1 \path [draw=cbfp3!\bval!lightgray,ultra thick] (t) -- +(0,0.4) (b) -- +(0,-0.4); \fi
\ifnum#2=2 \path [draw=cbfp3!\bval!lightgray,ultra thick] (l) -- +(-0.4,0) (r) -- +(0.4,0); \fi
\ifnum#2=3 \path [draw=cbfp3!\bval!lightgray,ultra thick] (t) -- +(0,0.4) (b) -- +(0,-0.4) (l) -- +(-0.3,0) (r) -- +(0.3,0); \fi
\ifnum#2=4 \path [draw=cbfp3!\bval!lightgray,ultra thick] (r) -- +(0.4,0) (b) -- +(0,-0.4); \fi
\ifnum#2=5 \path [draw=cbfp3!\bval!lightgray,ultra thick] (t) -- +(0,0.4) (l) -- +(-0.4,0); \fi

\ifnum#3=1 
\path [ultra thick, draw=cbfp2] (l) edge node [font=\sffamily\tiny, sloped, midway, inner sep=0.5pt, minimum size=0.2pt, circle, align=center, fill=white,opacity=1] {\textcolor{cbfp2}{$1-x$}} (t);
\fi

\ifnum#4=1 
\path [ultra thick, draw=cbfp5] (tr) edge node [font=\sffamily\tiny, sloped, midway, inner sep=0.5pt, minimum size=0.2pt, circle, align=center, fill=white,opacity=1] {\textcolor{cbfp5}{$-1$}} (bl); 
\fi

\ifnum#5=1 \path [ultra thick] (b) edge node [] {} (r); \fi
\ifnum#6=1 \path [ultra thick] (t) edge node [] {} (tr); \fi
\ifnum#7=1 \path [ultra thick] (tr) edge node [] {} (r); \fi
\ifnum#8=1 \path [ultra thick] (bl) edge node [] {} (l); \fi
\ifnum#9=1 \path [ultra thick] (b) edge node [] {} (bl); \fi

\end{scope}
}

\draw [black, draw opacity=0.5] (-7.2,1.0) rectangle (1.2,-1.4);
\draw [black, draw opacity=0.5] ( 1.8,1.0) rectangle (7.2,-1.4);
\draw [black, draw opacity=0.5] ( 1.8,-2.6) rectangle (7.2,-5.0);

\node [] () at (-3.0,-1.7) {{empty\vphantom{fhligpy}} \gtE};
\node [] () at ( 4.5,-1.7) {{marshaller\vphantom{fhligpy}} \gtF};
\node [] () at ( 4.5,-5.3) {$1001$\texttt{\vphantom{fhligpy}} \gtI};
\node [] () at (-6.0,-5.3) {{vertical\vphantom{fhligpy}} \gtV};
\node [] () at (-3.0,-5.3) {{horizontal\vphantom{fhligpy}} \gtH};
\node [] () at ( 0.0,-5.3) {{crossing\vphantom{fhligpy}} \gtC};

\gadgetstate{-6.0,-0.2}{0}{1}{1}{1}{0}{0}{0}{0}
\gadgetstate{-3.0,-0.2}{0}{1}{0}{0}{0}{1}{0}{1}
\gadgetstate{-0.0,-0.2}{0}{0}{0}{1}{1}{0}{1}{0}

\gadgetstate{ 3.0,-0.2}{4}{1}{1}{0}{0}{0}{0}{0}
\gadgetstate{ 6.0,-0.2}{4}{0}{0}{0}{1}{0}{1}{0}

\gadgetstate{-6.0,-3.8}{1}{0}{0}{0}{0}{1}{1}{0}
\gadgetstate{-3.0,-3.8}{2}{0}{0}{0}{1}{0}{0}{1}
\gadgetstate{ 0.0,-3.8}{3}{0}{1}{0}{0}{0}{0}{0}

\gadgetstate{ 3.0,-3.8}{5}{0}{1}{1}{0}{0}{0}{0}
\gadgetstate{ 6.0,-3.8}{5}{0}{0}{0}{0}{1}{0}{1}

\end{tikzpicture}
\caption{Computing the weight of the $H$-gadget in different states. Cases not shown here have no internal perfect matching. }
\label{fig:matchgate-proof}

\textcolor{lightgray}{\hrule}
\end{figure}
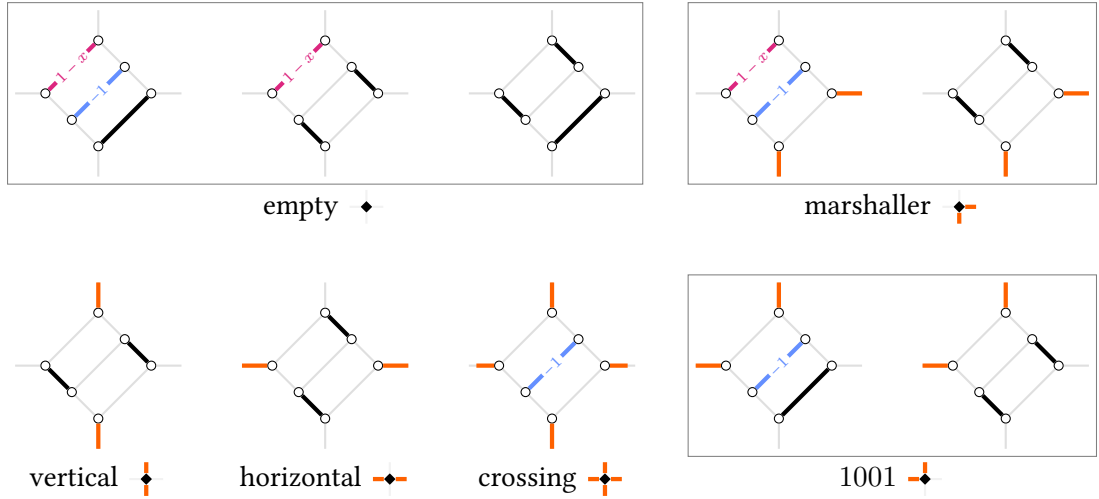

\begin{proof}
For assignments $a$ of odd Hamming weight, the graph $H-S$ has an odd number of vertices and thus no perfect matching, so the weight of $H$ under $a$ is zero.
This leaves $8$ cases to consider. Among these, the cases $1100$ \gtITR{} and $0011$ \gtIBL{} also admit no perfect matching in $H-S$. 
For the remaining $6$ cases, \Cref{fig:matchgate-proof} enumerates all the possible matchings and their weights. Notably, in the invalid state $1001$~\gtI, the two perfect matchings have weights $-1$ and $1$, and they hence cancel.
\end{proof}

To simplify the construction, we introduce two ``sentinel'' gadgets to be placed at the grid border. They are not necessary, but they allow us to state the proof in a uniform way. 

\begin{itemize}
\item The \emph{single-edge} gadget contains vertices $u,v$ connected by an edge, with a dangling edge at $v$. 
\item The \emph{single-vertex} gadget contains a single vertex $v$ with a dangling edge.
\end{itemize}
Both gadgets can be easily removed from the construction without changing the number of perfect matchings. 

\subsection{Graph construction}

Next, we construct an edge-weighted graph $G_{\bm{X}}$ by the following recipe. See \Cref{fig:per-from-grid} for the outcome when $\bm{X}$ is a $4$-by-$4$ matrix.
\begin{enumerate}
\item For all $i,j \in [m]$, introduce a fresh $H$-copy $H_{i,j}$. Replace the indeterminate in $H_{i,j}$ by $x_{ij}$. 
\item For all $i\in [m]$ and $j\in [m-1]$, join the $r$-edge of $H_{i,j}$ with the $\ell$-edge of $H_{i,j+1}$. 
\item For all $i\in [m-1]$ and $j\in [m]$, join the $b$-edge of $H_{i,j}$ with the $t$-edge of $H_{i+1,j}$. 
\item For all $i\in [m]$, introduce a fresh single-edge gadget, and join its dangling edge with the $\ell$-edge of $H_{i,1}$. 
Perform the same for $t$-edge of $H_{1,j}$, for all $j\in [m]$.  
\item For all $i\in [m]$, introduce a fresh single-vertex gadget, and join the dangling edge with the $r$-edge of $H_{i,m}$. 
Perform the same for the $b$-edge of $H_{m,j}$, for all $j\in [m]$.  
\end{enumerate}

Consequently, each row and column contains $m+1$ edges, all unweighted, that are formed by joining two dangling edges. 
Let $\tau$ be a $0$-$1$ assignment to all of these edges; it determines the states of all $H$-gadgets. 
Under this assignment, the weighted perfect matching count $\wt(\tau) \coloneqq \npm(G_{\bm{X}}\mid\tau)$ is the product of the weights of all gadgets under their respective local sub-assignments of $\tau$.

\subsection{Analysis of the construction}
In the following, denote the set of all non-vanishing assignments $\tau$ by 
\[
\Gamma\defeq\{\tau ~\colon~ \wt(\tau)\not\equiv 0\}.
\]
These non-vanishing assignments are characterised by the following series of propositions. 

\begin{proposition} \label{prop:flip}
Let $\tau\in\Gamma$ be non-vanishing and $i\in[m]$. 
For any $j\in[m]$, let $s_j$ be the assignment of $\tau$ to the $r$-edge of $H_{i,j}$, and let $s_0$ be the assignment of $\tau$ to the $\ell$-edge of $H_{i,1}$. 
Then there exists $j^*$ in $[m]$ such that $s_j=0$ for all $j<j^*$ and $s_j=1$ for all $j\geq j^*$. 
\end{proposition}

\begin{proof}
No gadget is invalid since $\tau\in\Gamma$. 
This means there is no such $j$ that $s_j=1$ while $s_{j+1}=0$. 
Further, the single-edge gadget on the $i$-th row ensures that $s_0=0$, and the single-vertex gadget on the $i$-th row ensures that $s_m=1$. 
This leaves the patterns in the proposition the only possibilities. 
\end{proof}

A column version of \Cref{prop:flip} holds analogously. 
Therefore, each row and column has exactly one \gtF{} gadget.
Furthermore, once the positions of all these \gtF{} gadgets are fixed, the only non-vanishing mapping $\tau$ is also determined. 
This yields the following bijection. 
\begin{corollary} \label{cor:bijection}
There is a bijection $h:\Gamma\to S_m$. 
If $h(\tau)=\pi$, then $H_{i,\pi i}$ is the \gtF{} gadget of $\tau$ on the $i$-th row, for any $i\in[m]$.
\end{corollary}

It is left for us to compute $\wt(\tau)$. 
In the following, recall that an ordered pair $i,j$ forms an inversion if $i<j$ and $\pi i>\pi j$. 

\begin{lemma}\label{lem:inversions}
Let $\tau\in\Gamma$ and $\pi=h(\tau)$. The number of inversions of $\pi$ equals the number of \gtC{} gadgets in $\tau$.
\end{lemma}

\begin{figure}[t]
\centering
\begin{tikzpicture}
\clip (-7,-1.6) rectangle (7,1.6);
\draw [black,opacity=0] (-7,-1.6) rectangle (7,1.6);

\begin{scope}[shift={(-3.5,0)}]
    
\node [diamond, draw=black, inner sep=2pt, fill=cbfp1!50] (g1i) at (-1.75,0) {};
\node [diamond, draw=black, inner sep=2pt, fill=cbfp1!50] (g1j) at (0,1.25) {};
\node [diamond, draw=black, inner sep=2pt] (g1c) at (0,0) {};
\node [rectangle, draw=black, inner sep=2pt] (g1ix) at (-1.75,-1) {};
\node [rectangle, draw=black, inner sep=2pt] (g1jx) at (0,-1) {};
\node [rectangle, draw=black, inner sep=2pt] (g1iy) at (1.25,0) {};
\node [rectangle, draw=black, inner sep=2pt] (g1jy) at (1.25,1.25) {};

\path [draw=cbfp3!, very thick] (g1ix) -- (g1i) -- (g1c) -- (g1iy);
\path [draw=cbfp3!, very thick] (g1jx) -- (g1c) -- (g1j) -- (g1jy);

\node [below=0 of g1ix] (t1) {\small $i$};
\node [below=0 of g1jx] (t2) {\small $j$};
\node [right=0 of g1iy] (t3) {\small $\pi i$};
\node [right=0 of g1jy] (t4) {\small $\pi j$};
\node () at ($(t1)!0.5!(t2)$) {\small $<$};
\node () at ($(t3)!0.5!(t4)$) {\small $\land$};
\node [above left=-0.125 of g1c] {\scriptsize $(j,\pi i)$};

\end{scope}

\begin{scope}[shift={(3.5,0)}]
    
\node [diamond, draw=black, inner sep=2pt, fill=cbfp1!50] (g2i) at (-1.75,1.25) {};
\node [diamond, draw=black, inner sep=2pt, fill=cbfp1!50] (g2j) at (0,0) {};
\node [rectangle, draw=black, inner sep=2pt] (g2ix) at (-1.75,-1) {};
\node [rectangle, draw=black, inner sep=2pt] (g2jx) at (0,-1) {};
\node [rectangle, draw=black, inner sep=2pt] (g2iy) at (1.25,1.25) {};
\node [rectangle, draw=black, inner sep=2pt] (g2jy) at (1.25,0) {};

\path [draw=cbfp3!, very thick] (g2ix) -- (g2i) -- (g2iy);
\path [draw=cbfp3!, very thick] (g2jx) -- (g2j) -- (g2jy);

\node [below=0 of g2ix] (s1) {\small $i$};
\node [below=0 of g2jx] (s2) {\small $j$};
\node [right=0 of g2iy] (s3) {\small $\pi i$};
\node [right=0 of g2jy] (s4) {\small $\pi j$};
\node () at ($(s1)!0.5!(s2)$) {\small $<$};
\node () at ($(s3)!0.5!(s4)$) {\small $\land$};

\end{scope}

\end{tikzpicture}
\caption[]{The gadget $H_{j,\pi i}$ is in state \gtC{} iff $(i,j)$ forms an inversion, i.e., $i<j$ and $\pi i>\pi j$. The left part shows an inversion, the right part shows a non-inversion.}

\label{fig:inversion}

\textcolor{lightgray}{\hrule}
\end{figure}
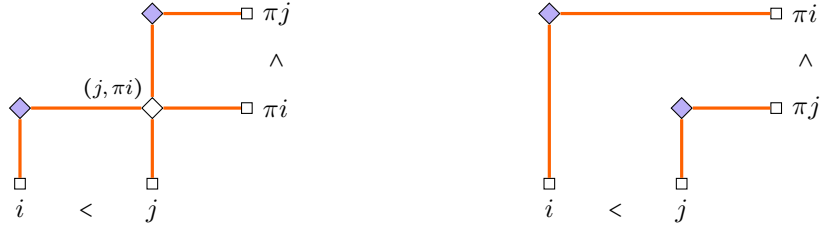

\begin{proof}
The mapping $\varphi: (i,j)\mapsto(j,\pi i)$ clearly is a bijection $[m]^2 \to [m]^2$. 
We show that $\varphi$ is a bijection between inversions in $\pi$ and \gtC{} gadgets in $\tau$. 
(See also \Cref{fig:inversion} for a quick illustration.)

\begin{itemize}
\setlength\itemsep{0em}
\item Suppose $i,j$ forms an inversion. 
Then $i<j$ and the column indices of the \gtF{} gadgets on row $i$ and $j$ satisfy $\pi i>\pi j$, leaving $H_{j,\pi i}$ in state \gtC{}. 
\item Conversely, suppose $H_{j,\pi i}$ is a \gtC{} gadget. 
The \gtF{} gadget on row $j$ lies at position $(j,\pi j)$, which implies $\pi i>\pi j$. 
The \gtF{} gadget on column $\pi i$ is at $(\pi^{-1}\pi i,\pi i)=(i,\pi i)$, and therefore $i<j$. 
Hence $(i,j)$ forms an inversion. \qedhere
\end{itemize}
\end{proof}

\subsection{Proof of main theorems}

All parts are now ready to prove the main theorems.

\begin{proof}[Proof of~\Cref{lem:reduction-grid}]
For any $\tau\in\Gamma$, using \Cref{lem:gadget-weight,lem:inversions} to compute the weight, we have
\[
\wt(\tau)\equiv (-1)^{\inv(\pi)}\prod_{i=1}^{m}x_{i, \pi i}
\]
with the bijection $\pi=h(\tau)$ from \Cref{cor:bijection}. 
Since only assignments $\tau \in \Gamma$ contribute to $\npm(G_{\bm{X}})$, 
\[
\npm(G_{\bm{X}})\equiv\sum_{\tau\in\Gamma}\wt(\tau)\equiv\sum_{\pi\in S_m}(-1)^{\inv(\pi)}\prod_{i=1}^{m}x_{i, \pi i}\equiv\det(\bm{X}). 
\]
Finally, the graph $G_{\bm{X}}$ is indeed the subgraph of an $3m\times 5m$ grid.
By assigning weight $0$ to all unused grid edges, $G_{\bm X}$ may also be viewed as an edge-weighted $3m\times 5m$ grid graph.
\Cref{fig:grid-embedding} depicts an exemplary grid embedding of $G_{\bm{X}}$ up to single-edge gadgets.
\end{proof}

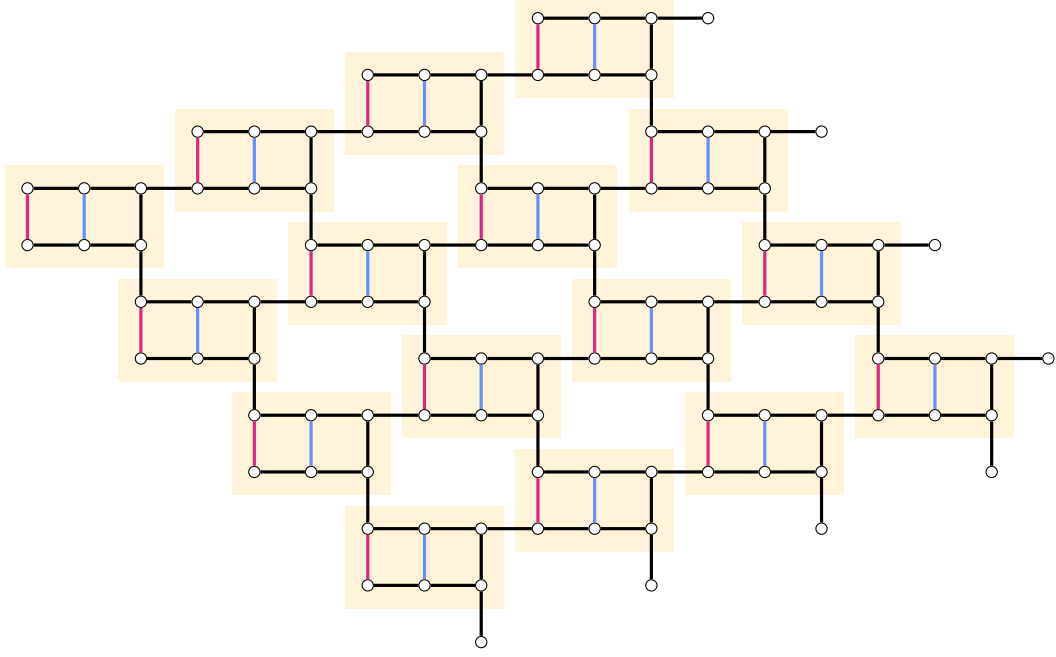
\begin{figure}[t]
\centering
\begin{tikzpicture}
\clip (-8,1) rectangle (8,9.8);
\draw [black,opacity=0] (-8,1) rectangle (8,9.8);

\tikzset{dedge/.style={font=\ttfamily\small, pos=0.67, inner sep=0.5pt, minimum size=0.2pt, circle, align=center, fill=white,opacity=1}}

\newcommand{\gridgadget}[4]{
\draw [draw opacity=0, fill=cbfp4!15!white] ({#1-1.5*#3+0.1},{#2-1*#3+0.1}) rectangle ({#1+1.5*#3-0.1},{#2+1*#3-0.1});
\foreach \x in {1,2,3} {
\foreach \y in {1,2} {
\node [circle, draw=black, fill=white, inner sep=1.5pt] (\x\y#4) at ({(#1+#3*(-2+\x))},{(#2+#3*(-1.5+\y))}) {};
}
}
\path [very thick, draw=black] (11#4) -- (21#4) -- (31#4) -- (32#4) -- (22#4) -- (12#4);
\path [very thick, draw=cbfp2] (11#4) edge node [] {} (12#4);
\path [very thick, draw=cbfp5] (21#4) edge node [] {} (22#4);
}

\begin{scope}[shift={(-2,0)},scale=0.75]
\foreach \xx in {4,3,2,1}{
\foreach \yy in {1,2,3,4}{
    \pgfmathparse{int(3*\xx-2*\yy)} \xdef\xxx{\pgfmathresult}
    \pgfmathparse{int(\xx+2*\yy)} \xdef\yyy{\pgfmathresult}
    \gridgadget{\xxx}{\yyy}{1}{a\xx\yy}
    
    \ifthenelse{\xx = 4}{
        \node [circle, draw=black, fill=white, inner sep=1.5pt] () at ($(32a\xx\yy)+(1,0)$) {} edge [very thick, draw=black] (32a\xx\yy);
    }{
        \pgfmathparse{int(\xx+1)} \xdef\xo{\pgfmathresult}
        \path [very thick, draw=black] (32a\xx\yy) -- (11a\xo\yy);;
    }

    \ifthenelse{\yy = 1}{
        \node [circle, draw=black, fill=white, inner sep=1.5pt] () at ($(31a\xx\yy)+(0,-1)$) {} edge [very thick, draw=black] (31a\xx\yy);
    }{
        \pgfmathparse{int(\yy-1)} \xdef\yo{\pgfmathresult}
        \path [very thick, draw=black] (31a\xx\yy) -- (12a\xx\yo);;
    }

}
}

\end{scope}
\end{tikzpicture}
\caption{A grid embedding of the graph $G_{\bm{X}}$, where $\bm{X}$ is a $4$-by-$4$ matrix.}
\label{fig:grid-embedding}

\textcolor{lightgray}{\hrule}
\end{figure}

\begin{proof}[Proof of~\Cref{thm:main}]
Assume $\omega=\omegadet > 2$, as the statement holds vacuously otherwise. 
Suppose there is an algorithm $\mathcal{A}$ that solves \probname{\#PlanarPM} in time $O(n^{\omega/2-\epsilon})$ for $\epsilon>0$.
Given $A\in F^{m\times m}$, construct the graph $G_A$ obtained from \Cref{lem:reduction-grid} by substituting $x_{ij}=a_{ij}$.
Then $\npm(G_A)=\det(A)$, and running $\mathcal{A}$ on $G_A$ outputs $\det(A)$ in time $O(m^{\omega-2\epsilon})$, contradicting the definition of $\omega$. 
The graph $G_{A}$ is a grid subgraph; see \Cref{fig:grid-embedding}.
\end{proof}

To show \Cref{thm:main-modulo}, we replace the edges of weights $1-x_{ij}$ or $-1$ with the gadget $M_k$ shown in \Cref{fig:mod-gadget}. 
The following proposition is immediate.

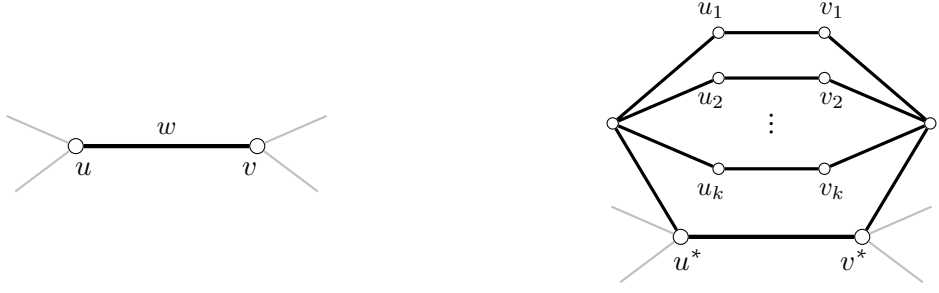
\begin{figure}[t]
\begin{tikzpicture}
\tikzset{dnode/.style={circle, draw=black, fill=white, inner sep=1.5pt}}
\clip (-8,-2.15) rectangle (8,1.75);
\draw [black,opacity=0] (-8,-2.15) rectangle (8,1.75);

\begin{scope}[shift={(-4,1.2)}]
\node [circle, draw=black, fill=white, inner sep=2pt] (l4) at (-1.2,-1.5) {};
\node [circle, draw=black, fill=white, inner sep=2pt] (r4) at ( 1.2,-1.5) {};
\path [ultra thick, draw=black] (l4) edge node [above] {\small $w$} (r4);

\node [] at (-1.1,-1.8) {$u$}; 
\node [] at (1.1,-1.8) {$v$}; 
                                
\path [draw=lightgray, thick] 
    (l4) -- +(-0.8,-0.6)
    (l4) -- +(-0.917,0.4)
    (r4) -- +(0.8,-0.6)
    (r4) -- +(0.917,0.4)
;
\end{scope}

\begin{scope}[shift={(4,0)}]
\node [dnode] (l) at (-2.1,0) {}; 
\node [dnode] (r) at ( 2.1,0) {};
\node [dnode] (l1) at (-0.7,1.2) {};
\node [dnode] (r1) at ( 0.7,1.2) {};
\node [dnode] (l2) at (-0.7,0.6) {};
\node [dnode] (r2) at ( 0.7,0.6) {};
\node [dnode] (l3) at (-0.7,-0.6) {};
\node [dnode] (r3) at ( 0.7,-0.6) {};
\node [circle, draw=black, fill=white, inner sep=2pt] (l4) at (-1.2,-1.5) {};
\node [circle, draw=black, fill=white, inner sep=2pt] (r4) at ( 1.2,-1.5) {};
\path [very thick, draw=black]  (l) -- (l1) -- (r1) -- (r)
                                (l) -- (l2) -- (r2) -- (r)
                                (l) -- (l3) -- (r3) -- (r)
                                (l) -- (l4)
                                (r4) -- (r);
\path [ultra thick, draw=black] (l4) -- (r4);
                                
\path [draw=lightgray, thick] 
    (l4) -- +(-0.8,-0.6)
    (l4) -- +(-0.917,0.4)
    (r4) -- +(0.8,-0.6)
    (r4) -- +(0.917,0.4)
;

\node [] at (0,0) {$\vdots$}; 
\node [] at (-1.1,-1.8) {$u^*$}; 
\node [] at (1.1,-1.8) {$v^*$}; 

\node [] at (-0.8,1.5) {\small $u_1$}; 
\node [] at (-0.8,0.3) {\small $u_2$}; 
\node [] at (-0.8,-0.9) {\small $u_k$}; 
\node [] at (0.8,1.5) {\small $v_1$}; 
\node [] at (0.8,0.3) {\small $v_2$}; 
\node [] at (0.8,-0.9) {\small $v_k$}; 
\end{scope}

\end{tikzpicture}
\caption{Replacing an edge $uv$ of weight $w\neq 1$ with gadget $M_k$ for $k=(w-1)^{-1}$ over $\mathbb{F}_p$.}
\label{fig:mod-gadget}

\textcolor{lightgray}{\hrule}
\end{figure}

\begin{proposition} \label{prop:modulo-gadget}
The graph $M_k$ has $k+1$ perfect matchings when both dangling edges are assigned $0$, and $k$ perfect matchings when both dangling edges are assigned $1$. 
\end{proposition}

\begin{proof}[Proof of~\Cref{thm:main-modulo}]
The theorem is trivial over $\mathbb{F}_2$, so assume $p\geq 3$. 
All equations are taken over $\mathbb{F}_p$. 

First, replace all indeterminates in $G_{\bm{X}}$ with the concrete values from the input matrix. 
Then, for any edge $e=(u,v)$ with weight $w(e)\notin\{0,1\}$ in $G_{\bm{X}}$, 
introduce a fresh gadget $M_k$ with $k=(w(e)-1)^{-1}\in\{1,\dots,p-1\}$, remove the original edge, and identify $u$ with $u^*$, $v$ with $v^*$. 
Call the new graph $G^*$. 
By \Cref{prop:modulo-gadget}, each original edge with weight $w(e)$ ($w(e)\notin\{0,1\}$) contributes a factor of $k$ when not selected in the perfect matching, and $k+1=k\cdot w(e)$ when selected in the perfect matching. 
Combining with \Cref{lem:reduction-grid}, we have
\[
\npm(G^*)=\left(\prod_{\substack{e\in E(G_{\bm{X}})\\w(e)\notin\{0,1\}}}(w(e)-1)^{-1}\right)\cdot\npm(G_{\bm{X}})=\left(\prod_{\substack{e\in E(G_{\bm{X}})\\w(e)\notin\{0,1\}}}(w(e)-1)^{-1}\right)\cdot\det(\bm{X}). 
\]

The factor in the above equation is non-zero modulo $p$. 
The number of vertices of the new graph $G^*$ blows up by a factor of at most $3p$. 
It takes time $O(|V(G)|\cdot p)$ to construct the graph $G^*$ and compute the factor. 
This turns any algorithm that computes $\npm(G^*)\bmod{p}$ in time $O(n^{\omega/2-\epsilon})$ into an algorithm that computes $\det(\bm{X})$ in time $O(m^{\omega-2\epsilon})$, a contradiction.
\end{proof}

\bibliographystyle{alpha}
\bibliography{refs.bib}

\newcommand{\etalchar}[1]{$^{#1}$}
\begin{thebibliography}{CFGW22}

\bibitem[Bac21]{DBLP:journals/siamcomp/Backens21}
Miriam Backens.
\newblock A full dichotomy for {$\textsf{Holant}^{\mathbf{c}}$}, inspired by quantum computation.
\newblock {\em {SIAM} J. Comput.}, 50(6):1739--1799, 2021.

\bibitem[Bax82]{Baxter82}
Rodney~J. Baxter.
\newblock {\em Exactly Solved Models in Statistical Mechanics}.
\newblock Academic press limited {London}, 1982.

\bibitem[BCK{\etalchar{+}}26]{DBLP:journals/eccc/BrandCKLOS026}
Cornelius Brand, Radu Curticapean, Petteri Kaski, Baitian Li, Ian Orzel, Tim Seppelt, and Jiaheng Wang.
\newblock Beyond bilinear complexity: What works and what breaks with many modes?
\newblock {\em Electron. Colloquium Comput. Complex.}, {TR26-025}, 2026.
\newblock To appear in CCC 2026.

\bibitem[Ber84]{DBLP:journals/ipl/Berkowitz84}
Stuart~J. Berkowitz.
\newblock On computing the determinant in small parallel time using a small number of processors.
\newblock {\em Inf. Process. Lett.}, 18(3):147--150, 1984.

\bibitem[BH74]{MR331751}
James~R. Bunch and John~E. Hopcroft.
\newblock Triangular factorization and inversion by fast matrix multiplication.
\newblock {\em Math. Comp.}, 28:231--236, 1974.

\bibitem[BS83]{DBLP:journals/tcs/BaurS83}
Walter Baur and Volker Strassen.
\newblock The complexity of partial derivatives.
\newblock {\em Theor. Comput. Sci.}, 22:317--330, 1983.

\bibitem[B{\"{u}}r24]{DBLP:journals/corr/abs-2406-06217}
Peter B{\"{u}}rgisser.
\newblock Completeness classes in algebraic complexity theory.
\newblock {\em CoRR}, abs/2406.06217, 2024.

\bibitem[CC07]{DBLP:journals/ijsi/CaiC07}
Jin{-}yi Cai and Vinay Choudhary.
\newblock Some results on matchgates and holographic algorithms.
\newblock {\em Int. J. Softw. Informatics}, 1(1):3--36, 2007.

\bibitem[CC17]{CaiChenBook}
Jin-Yi Cai and Xi~Chen.
\newblock {\em Complexity dichotomies for counting problems. {V}ol. 1}.
\newblock Cambridge University Press, Cambridge, 2017.
\newblock Boolean domain.

\bibitem[CCL09]{DBLP:journals/mst/CaiCL09}
Jin{-}yi Cai, Vinay Choudhary, and Pinyan Lu.
\newblock On the theory of matchgate computations.
\newblock {\em Theory Comput. Syst.}, 45(1):108--132, 2009.

\bibitem[CF22]{DBLP:journals/siamcomp/CaiF22}
Jin{-}Yi Cai and Zhiguo Fu.
\newblock Holographic algorithm with matchgates is universal for planar {\#}csp over boolean domain.
\newblock {\em {SIAM} J. Comput.}, 51(2):17--50, 2022.

\bibitem[CFGW22]{DBLP:journals/mst/CaiFGW22}
Jin{-}Yi Cai, Zhiguo Fu, Heng Guo, and Tyson Williams.
\newblock {FKT} is not universal - {A} planar {H}olant dichotomy for symmetric constraints.
\newblock {\em Theory Comput. Syst.}, 66(1):143--308, 2022.

\bibitem[CG14]{DBLP:journals/toc/CaiG14}
Jin{-}Yi Cai and Aaron Gorenstein.
\newblock Matchgates revisited.
\newblock {\em Theory Comput.}, 10:167--197, 2014.

\bibitem[CL11]{DBLP:journals/jcss/CaiL11}
Jin{-}yi Cai and Pinyan Lu.
\newblock Holographic algorithms: From art to science.
\newblock {\em J. Comput. Syst. Sci.}, 77(1):41--61, 2011.

\bibitem[CLS{\etalchar{+}}24]{CLSWW2024}
Matthias Christandl, Vladimir Lysikov, Vincent Steffan, Albert~H. Werner, and Freek Witteveen.
\newblock The resource theory of tensor networks.
\newblock {\em {Quantum}}, 8:1560, December 2024.

\bibitem[CLX20]{DBLP:journals/mst/CaiLX20}
Jin{-}Yi Cai, Pinyan Lu, and Mingji Xia.
\newblock Dichotomy for {H}olant\({}^{\mbox{{\({_\ast}\)}}}\) problems on the boolean domain.
\newblock {\em Theory Comput. Syst.}, 64(8):1362--1391, 2020.

\bibitem[CX22]{DBLP:conf/sosa/CurticapeanX22}
Radu Curticapean and Mingji Xia.
\newblock Parameterizing the permanent: Hardness for fixed excluded minors.
\newblock In Karl Bringmann and Timothy~M. Chan, editors, {\em 5th Symposium on Simplicity in Algorithms, SOSA@SODA 2022, Virtual Conference, January 10-11, 2022}, pages 297--307. {SIAM}, 2022.

\bibitem[Edm65]{MR177907}
Jack Edmonds.
\newblock Paths, trees, and flowers.
\newblock {\em Canadian J. Math.}, 17:449--467, 1965.

\bibitem[EV15]{PhysRevLett.115.180405}
G.~Evenbly and G.~Vidal.
\newblock Tensor network renormalization.
\newblock {\em Phys. Rev. Lett.}, 115:180405, Oct 2015.

\bibitem[FKL07]{DBLP:conf/isaac/FlarupKL07}
Uffe Flarup, Pascal Koiran, and Laurent Lyaudet.
\newblock On the expressive power of planar perfect matching and permanents of bounded treewidth matrices.
\newblock In Takeshi Tokuyama, editor, {\em Algorithms and Computation, 18th International Symposium, {ISAAC} 2007, Sendai, Japan, December 17-19, 2007, Proceedings}, volume 4835 of {\em Lecture Notes in Computer Science}, pages 124--136. Springer, 2007.

\bibitem[Geo73]{MR388756}
Alan George.
\newblock Nested dissection of a regular finite element mesh.
\newblock {\em SIAM J. Numer. Anal.}, 10:345--363, 1973.

\bibitem[HL16]{DBLP:journals/cc/HuangL16}
Sangxia Huang and Pinyan Lu.
\newblock A dichotomy for real weighted {H}olant problems.
\newblock {\em Comput. Complex.}, 25(1):255--304, 2016.

\bibitem[HPM{\etalchar{+}}19]{Huggins_2019}
William Huggins, Piyush Patil, Bradley Mitchell, K~Birgitta Whaley, and E~Miles Stoudenmire.
\newblock Towards quantum machine learning with tensor networks.
\newblock {\em Quantum Science and Technology}, 4(2):024001, Jan 2019.

\bibitem[IP01]{ImpagliazzoP01}
Russell Impagliazzo and Ramamohan Paturi.
\newblock On the complexity of k-{SAT}.
\newblock {\em J. Comput. Syst. Sci.}, 62(2):367--375, 2001.

\bibitem[JVW90]{MR1049758}
F.~Jaeger, D.~L. Vertigan, and D.~J.~A. Welsh.
\newblock On the computational complexity of the {J}ones and {T}utte polynomials.
\newblock {\em Math. Proc. Cambridge Philos. Soc.}, 108(1):35--53, 1990.

\bibitem[Kas63]{MR153427}
P.~W. Kasteleyn.
\newblock Dimer statistics and phase transitions.
\newblock {\em J. Mathematical Phys.}, 4:287--293, 1963.

\bibitem[KK08]{DBLP:conf/issac/KaltofenK08}
Erich~L. Kaltofen and Pascal Koiran.
\newblock Expressing a fraction of two determinants as a determinant.
\newblock In J.~Rafael Sendra and Laureano Gonz{\'{a}}lez{-}Vega, editors, {\em Symbolic and Algebraic Computation, International Symposium, {ISSAC} 2008, Linz/Hagenberg, Austria, July 20-23, 2008, Proceedings}, pages 141--146. {ACM}, 2008.

\bibitem[LRT79]{MR526496}
Richard~J. Lipton, Donald~J. Rose, and Robert~Endre Tarjan.
\newblock Generalized nested dissection.
\newblock {\em SIAM J. Numer. Anal.}, 16(2):346--358, 1979.

\bibitem[MP08]{DBLP:journals/jc/MalodP08}
Guillaume Malod and Natacha Portier.
\newblock Characterizing {V}aliant's algebraic complexity classes.
\newblock {\em J. Complex.}, 24(1):16--38, 2008.

\bibitem[MS06]{DBLP:journals/algorithmica/MuchaS06}
Marcin Mucha and Piotr Sankowski.
\newblock Maximum matchings in planar graphs via gaussian elimination.
\newblock {\em Algorithmica}, 45(1):3--20, 2006.

\bibitem[MV97]{DBLP:journals/cjtcs/MahajanV97}
Meena Mahajan and V.~Vinay.
\newblock Determinant: Combinatorics, algorithms, and complexity.
\newblock {\em Chic. J. Theor. Comput. Sci.}, 1997, 1997.

\bibitem[Sam42]{MR7497}
P.~A. Samuelson.
\newblock A method of determining explicitly the coefficients of the characteristic equation.
\newblock {\em Ann. Math. Statistics}, 13:424--429, 1942.

\bibitem[SC20]{DBLP:conf/focs/0001C20}
Shuai Shao and Jin{-}Yi Cai.
\newblock A dichotomy for real boolean {H}olant problems.
\newblock In Sandy Irani, editor, {\em 61st {IEEE} Annual Symposium on Foundations of Computer Science, {FOCS} 2020, Durham, NC, USA, November 16-19, 2020}, pages 1091--1102. {IEEE}, 2020.

\bibitem[Sch81]{DBLP:journals/siamcomp/Schonhage81}
Arnold Sch{\"{o}}nhage.
\newblock Partial and total matrix multiplication.
\newblock {\em {SIAM} J. Comput.}, 10(3):434--455, 1981.

\bibitem[Str69]{MR248973}
Volker Strassen.
\newblock Gaussian elimination is not optimal.
\newblock {\em Numer. Math.}, 13:354--356, 1969.

\bibitem[TF61]{MR136398}
H.~N.~V. Temperley and Michael~E. Fisher.
\newblock Dimer problem in statistical mechanics---an exact result.
\newblock {\em Philos. Mag. (8)}, 6:1061--1063, 1961.

\bibitem[Tod92]{toda1992classes}
Seinosuke Toda.
\newblock Classes of arithmetic circuits capturing the complexity of computing the determinant.
\newblock {\em IEICE Transactions on Information and Systems}, 75(1):116--124, 1992.

\bibitem[Val79a]{DBLP:conf/stoc/Valiant79a}
Leslie~G. Valiant.
\newblock Completeness classes in algebra.
\newblock In Michael~J. Fischer, Richard~A. DeMillo, Nancy~A. Lynch, Walter~A. Burkhard, and Alfred~V. Aho, editors, {\em Proceedings of the 11h Annual {ACM} Symposium on Theory of Computing, April 30 - May 2, 1979, Atlanta, Georgia, {USA}}, pages 249--261. {ACM}, 1979.

\bibitem[Val79b]{DBLP:journals/siamcomp/Valiant79}
Leslie~G. Valiant.
\newblock The complexity of enumeration and reliability problems.
\newblock {\em {SIAM} J. Comput.}, 8(3):410--421, 1979.

\bibitem[Val02a]{DBLP:journals/tcs/Valiant02}
Leslie~G. Valiant.
\newblock Expressiveness of matchgates.
\newblock {\em Theor. Comput. Sci.}, 289(1):457--471, 2002.

\bibitem[Val02b]{DBLP:journals/siamcomp/Valiant02}
Leslie~G. Valiant.
\newblock Quantum circuits that can be simulated classically in polynomial time.
\newblock {\em {SIAM} J. Comput.}, 31(4):1229--1254, 2002.

\bibitem[Val06]{DBLP:conf/focs/Valiant06}
Leslie~G. Valiant.
\newblock Accidental algorithms.
\newblock In {\em 47th Annual {IEEE} Symposium on Foundations of Computer Science, {FOCS} 2006, Berkeley, California, USA, October 21-24, 2006, Proceedings}, pages 509--517. {IEEE} Computer Society, 2006.

\bibitem[Val08]{DBLP:journals/siamcomp/Valiant08}
Leslie~G. Valiant.
\newblock Holographic algorithms.
\newblock {\em {SIAM} J. Comput.}, 37(5):1565--1594, 2008.

\bibitem[VW18]{WilliamsW18}
Virginia {Vassilevska Williams} and R.~Ryan Williams.
\newblock Subcubic equivalences between path, matrix, and triangle problems.
\newblock {\em J. {ACM}}, 65(5):27:1--27:38, 2018.

\bibitem[Wel93]{Welsh_1993}
Dominic Welsh.
\newblock {\em Complexity: Knots, Colourings and Countings}.
\newblock London Mathematical Society Lecture Note Series. Cambridge University Press, 1993.

\bibitem[Wil97]{DBLP:conf/soda/Wilson97}
David~Bruce Wilson.
\newblock Determinant algorithms for random planar structures.
\newblock In Michael~E. Saks, editor, {\em Proceedings of the Eighth Annual {ACM-SIAM} Symposium on Discrete Algorithms, 5-7 January 1997, New Orleans, Louisiana, {USA}}, pages 258--267. {ACM/SIAM}, 1997.

\bibitem[Wil18]{Williams18}
R.~Ryan Williams.
\newblock Faster all-pairs shortest paths via circuit complexity.
\newblock {\em {SIAM} J. Comput.}, 47(5):1965--1985, 2018.

\bibitem[Yus08]{DBLP:conf/focs/Yuster08}
Raphael Yuster.
\newblock Matrix sparsification for rank and determinant computations via nested dissection.
\newblock In {\em 49th Annual {IEEE} Symposium on Foundations of Computer Science, {FOCS} 2008, Philadelphia, PA, USA, October 25-28, 2008}, pages 137--145. {IEEE} Computer Society, 2008.

\bibitem[YZ07]{DBLP:conf/soda/YusterZ07}
Raphael Yuster and Uri Zwick.
\newblock Maximum matching in graphs with an excluded minor.
\newblock In Nikhil Bansal, Kirk Pruhs, and Clifford Stein, editors, {\em Proceedings of the Eighteenth Annual {ACM-SIAM} Symposium on Discrete Algorithms, {SODA} 2007, New Orleans, Louisiana, USA, January 7-9, 2007}, pages 108--117. {SIAM}, 2007.

\end{thebibliography}

\end{document}